\documentclass[11pt,conference,letterpaper,onecolumn]{IEEEtran}
\IEEEoverridecommandlockouts

\usepackage[T1]{fontenc}
\usepackage[utf8]{inputenc}
\usepackage{amsmath,amssymb,amsthm,mathtools}
\usepackage{wrapfig}
\usepackage{mathrsfs}
\usepackage{bm}
\usepackage{cite}
\usepackage{xspace}
\usepackage{color,xcolor}
\usepackage[top=1in,bottom=1in,left=0.75in,right=0.75in]{geometry}
\usepackage{microtype}
\definecolor{darkgreen}{rgb}{0,0.45,0}

\allowdisplaybreaks[2]

\makeatletter
\renewenvironment{abstract}{%
  \normalfont
  \par\addvspace{0.5\baselineskip}%
  \begin{center}%
    \@IEEEabskeysecsize\textbf{\abstractname}%
  \end{center}%
  \@IEEEabskeysecsize\noindent\ignorespaces
}{%
  \par\addvspace{0.5\baselineskip}%
  \normalfont\normalsize
}
\makeatother

\newtheorem{theorem}{Theorem}
\newtheorem{proposition}{Proposition}
\newtheorem{corollary}{Corollary}
\newtheorem{lemma}{Lemma}

\newtheorem{ex}{Ex}

\newcommand{\define}{\triangleq}
\newcommand{\defeq}{\mathrel{:=}}
\newcommand{\Ber}{\mathrm{Bern}}

\newcommand{\ket}[1]{|#1\rangle}

\newcommand{\ketbra}[2]{|#1\rangle\langle #2|}
\newcommand{\CalH}{\mathcal{H}}
\newcommand{\CalS}{\mathcal{S}}
\newcommand{\CalX}{\mathcal{X}}
\newcommand{\CalU}{\mathcal{U}}

\newcommand{\ScrR}{\mathscr{R}}

\newcommand{\hTwo}{h_2}
\newcommand{\Ced}{C_{\mathrm{ED}}}
\newcommand{\Cnc}{C_{\mathrm{nc}}^{\mathrm{sl}}}
\newcommand{\Cc}{C_{\mathrm{c}}}
\newcommand{\Cur}{\ScrR_{\mathrm{u,red}}}

\newcommand{\artanh}{\operatorname{artanh}}
\def\ttp{\mathtt{p}}
\def\olinekappa{\overline{\kappa}}
\def\complex{\mathbb{C}}
\def\ulineCalX{\underline{\CalX}}

\def\parsec{\par\noindent}
\def\med{\medskip\parsec}

\begin{document}

\title{\huge{Rate Loss in Quantum Channels with Classical State and Applications for Quantum Broadcast Channels}}

\author{%
  \IEEEauthorblockN{Igor BERNARD and Arun PADAKANDLA}
  \IEEEauthorblockA{EURECOM, France}
}

\maketitle

\begin{abstract}
We consider the problem of \textit{rate loss} - a strict penalty suffered in achievable rates due to the lack of channel state information at the receiver (Rx) of a classical-quantum (CQ) channel. First, we identify non-commutative CQ channels and analytically prove a rate loss. Building on this, we next prove that coset-code-based strategies can strictly outperform conventional unstructured IID-code-based strategies for non-commutative 3-user CQ broadcast channels.
\end{abstract}

\section{Introduction}
\label{Sec:Introduction}
The presence of auxiliary random variables (RVs) is ubiquitous in the characterization of rate regions in classical and quantum network information theory \cite{BkNITElGamalKim_2011,201512TIT_SavWil,202102SAD_Sen}. In essence, specification of inner and outer bounds requires several parameters. The auxiliary RVs, in addition to facilitating these parameters, provide an operational meaning that guides intuition in the design of coding strategies and outer bound techniques. However, the flip side is that these auxiliary RVs either leave the rate-region characterization without a compact description or render the resulting optimization problem intractable or infinite-dimensional. All of this makes it extremely hard to compare different inner bounds and definitively contrast the performance of different coding strategies.

\begin{wrapfigure}{l}{0.43\textwidth} \vspace{-0.15in} \centering \includegraphics[width=0.41\textwidth]{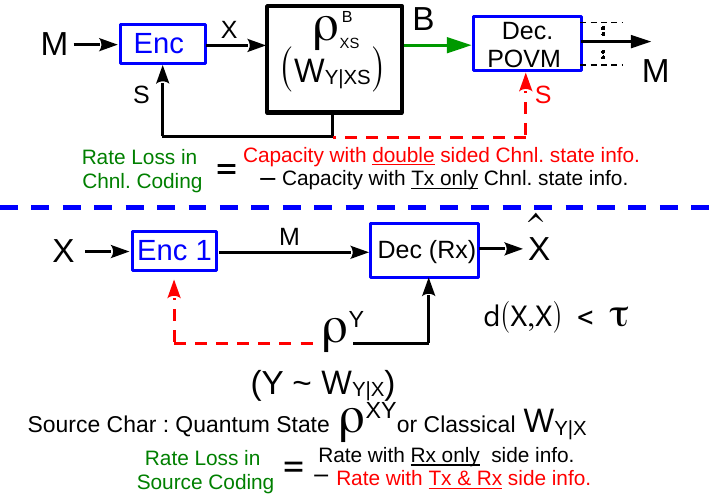} \vspace{-0.12in} \caption{A general CQ point-to-point (PTP) channel with Tx state information (QSTx). If the states $(\rho_{x,s}: (x,s) \in \CalX\times \CalS)$ commute, this reduces to a classical PTP channel with Tx state information \cite{1980MMPCIT_GelPin}. Bottom: PTP Rate-distortion with side info. Source could be classical-quantum (CQ) ($\rho^{XY}$) with classical system $X$ or purely classical with joint PMF $~\mathbb{W}_{XY}$.} \label{Fig:RateLoss} \vspace{-0.15in} \end{wrapfigure}

Among a handful of known techniques to address this difficulty, the notion of \textit{rate loss} (Fig.~\ref{Fig:RateLoss}) has proven to be quite useful in both source and channel coding. Consider a point-to-point (PTP) rate distortion problem wherein the receiver (Rx) has access to side information (SI) \cite{197601TIT_WynZiv}. Wyner and Ziv proved \cite[Sec.~II]{197601TIT_WynZiv} that \textit{unavailability} of SI at the Tx can result in a strict \textit{rate loss} \cite{199606TIT_Zam}. Wagner, Kelly and Altug \cite{201107TIT_WagKelAlt} crucially exploited Wyner and Ziv's rate loss to prove a consequential statement - sub-optimality of the quantize and bin (or Berger-Tung) strategy for the distributed rate distortion problem.

The dual notion of \textit{rate loss} in channel coding - the scenario of interest here, has also proven to be very useful. Consider communicating over a PTP channel whose transition depends on a random parameter \cite{1980MMPCIT_GelPin,201604JPA_BocCaiNot} referred to as a state (Fig.~\ref{Fig:RateLoss}). Dual to the above source coding problem, \textit{unavailability} of the state at the Rx can result in a \textit{rate loss} \cite{BkNITElGamalKim_2011}, i.e., lower achievable rates. Alternately stated, a state-cognizant Rx can squeeze out higher rates. A \textit{rate loss} statement in channel coding has enabled us to establish structural characterizations of optimal coding strategies. The need for a base layer with a common code for communicating statistically independent messages over a broadcast channel \cite{197905TIT_Mar,BkNITElGamalKim_2011} crucially relies on a rate loss in channel coding. The sub-optimality of conventional unstructured random codes for communicating over a multiple access channel with states is also proven \cite{200906TIT_PhiZam} via the rate loss.

While there have been several works \cite{199606TIT_Zam,200906TIT_PhiZam,200804TIT_CohZam,200408DIMACS_CohZam} for classical scenarios proving rate loss in both source and channel coding, characterization of a rate loss in truly noncommutative quantum channels (q-channels) remains unaddressed to the best of our knowledge. Our first finding (Prop.~\ref{Prop:QSTXStrictRateLoss}, Figs.~\ref{Fig:FigPlotRateLossWrtState}, \ref{Fig:FigPlotRateLossWrtCost}) fills this gap by identifying q-channels with classical state information over which bit communication entails a rate loss. Let us elaborate. {In Sec.~\ref{Sec:CQPTPSTx}, we focus on the PTP classical-quantum (CQ) channels with classical state information at the transmitter (Tx) (QSTx) (Fig.~\ref{Fig:RateLoss}). Boche, Cai, and Nötzel~\cite{201604JPA_BocCaiNot} characterized the capacity of QSTx when state information is available causally or non-causally. In particular, they provide a single-letter expression for the causal case. Leveraging the diagonal states as in the BB84 protocol, we prove a rate loss in the causal case. In the non-causal case, we show that the single-letter inner bound coincides with the causal capacity, while remaining strictly smaller than the capacity of the corresponding setting in which the Rx also has access to the channel state information.}

\begin{wrapfigure}{l}{0.43\textwidth} \vspace{-0.12in} \centering \includegraphics[width=0.41\textwidth]{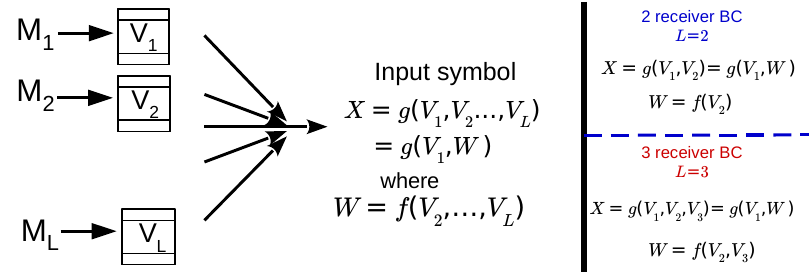} \vspace{-0.10in} \caption{\!\!\!From Rx $1$'s perspective, the interference $W$ it experiences is a specific combination of the other user codewords. While on $2$-user BC, $W$ is a univariate func.~of $V_{2}$, on a $3$-user BC, $W$ is a bivariate func.~of $V_{2},V_{3}$.} \label{Fig:InterferencOn3BC} \vspace{-0.15in} \end{wrapfigure}

Next, we build on the above findings in the context of a CQ broadcast channel (CQBC) \cite{201512TIT_SavWil,202606ISIT_GouPad,202503arXiv_GouPad}.  Let's briefly recall communication over a broadcast channel (BC). To communicate multiple bit streams to different Rxs over a BC, the Tx has to fuse chosen codewords from different codes through a single channel input (Fig.~\ref{Fig:InterferencOn3BC}). From the perspective of any Rx, a specific aggregation (denoted $W$ in Fig.~\ref{Fig:InterferencOn3BC}) of the codewords chosen for the other Rxs acts as interference. How does one mitigate interference? The two interference mitigation techniques \cite{BkNITElGamalKim_2011,201512TIT_SavWil} we are aware of are (i) precoding by Tx via Marton's binning \cite{197905TIT_Mar} and (ii) (partial or full) interference decoding by the Rx \cite{198101TIT_HanKob}. Precoding - the same technique employed by Tx in Fig.~\ref{Fig:RateLoss} - yields rate corresponding to that achievable in Fig.~\ref{Fig:RateLoss} with \underline{Tx only} state information, while the latter yields rates corresponding to Tx \underline{and} Rx state information. The presence of a \textit{rate loss} brings in an \underline{asymmetry} between these two techniques. An efficient coding strategy must therefore enable Rxs to decode as \textit{large} an interference component as possible, thereby suppressing the residual component precoded for by the Tx.

Let's focus these ideas for a $3$-user CQBC. In contrast to communication over a $2-$user BC, each Rx has to contend with \textit{two interferers}. The interference $W$ on a BC with $3$ Rxs is a \textit{bivariate} function of $V_{2},V_{3}$ - the signals of the other Rxs (Fig.~\ref{Fig:InterferencOn3BC}).  If the \textit{bivariate} function of $V_{2},V_{3}$ is strictly compressive, it is possible for Rxs over certain 3-CQBCs to be unable to decode the pair $V_{2},V_{3}$ even while being able to decode this bivariate function, since the latter is of a lower entropy. Over such $3$-CQBCs, a coding strategy \textit{unable} to decode bivariate functions effectively will precode for a larger component of the undecoded interference and suffer a rate loss. This connection leads us to our findings of Sec.~\ref{Sec:Broadcast}. In \cite{202606ISIT_GouPad,202503arXiv_GouPad}, a coset code based strategy that is specifically tailored to efficiently decode bivariate interference components efficiently is presented. Leveraging Ex.~\ref{Ex:QSTxExample} of Sec.~\ref{Sec:CQPTPSTx} as a central component and identifying a \textit{non-commutative} $3-$CQBC, we prove the coset code based strategy \cite{202606ISIT_GouPad,202503arXivBC_GouPad} achieves rate triples that are not achievable via any conventional unstructured IID code based strategy - a statement not established in \cite{202606ISIT_GouPad,202503arXivBC_GouPad}.

Our proof of the rate loss for QSTx in Ex.~\ref{Ex:QSTxExample} employs a convexity argument which does not rely on cardinality bounds for the auxiliary RV. The more involved step, needed for the broadcast-channel comparison, is to specialize the auxiliary RVs appearing in the unstructured IID code based inner bound \cite[Thm.~2]{201804TIT_PadPra}. This specialization is carried out in Sec.~\ref{Sec:Broadcast}.

\med\noindent\textbf{Notation.} We supplement standard notation with the following. All Hilbert spaces are finite-dimensional. For a Hilbert space $\CalH$, $\mathcal{L}(\CalH)$, $\mathcal{P}(\CalH)$ and $\mathcal{D}(\CalH)$ denote linear, positive and density operators acting on $\CalH$. For a finite set $\CalX$, $\mathcal{P}(\CalX)$ denotes the collection of PMFs on $\CalX$. For $p\in[0,1]$, $\hTwo(p)\define -p\log_2 p-(1-p)\log_2(1-p)$ denotes binary entropy and $p*q\define p(1-q)+(1-p)q$ denotes binary convolution. The Pauli matrices and diagonal states are denoted by
\[
X\!=\begin{pmatrix}\!0\!&1\\1\!&0\!\end{pmatrix}\!,
Y\!=\begin{pmatrix}\!0\!&-i\\ i\!&0\!\end{pmatrix}\!,
Z\!=\begin{pmatrix}\!1\!&0\\0\!&-1\!\end{pmatrix}\!, \ket{\pm}=\frac{\ket{0}\pm\ket{1}}{\sqrt{2}}.
\]
\section{{Rate Loss in CQ Channel with Classical State}}
\label{Sec:CQPTPSTx}

We begin with a formal description of a QSTx. Consider a (generic) \textit{QSTx} specified through (i) a finite input set $\CalX$, (ii) a finite set $\CalS$ of states, (iii) a PMF $\ttp_{S}(\cdot)$ on $\CalS$, (iv) a collection $(\rho_{xs} \in \mathcal{D}(\mathcal{H}): (x,s) \in \CalX\times \CalS )$ of density operators and (v) cost function $\kappa :\CalX \times \CalS \rightarrow [0,\infty)$. The cost function is additive, i.e., having observed the state sequence $s^{n}$ the cost incurred by the sender in preparing the state $\otimes_{t=1}^{n}\rho_{x_{t}s_{t}}$ is $\olinekappa(x^{n},s^{n}) \define \frac{1}{n}\sum_{t=1}^{n}\kappa(x_{t},s_{t})$. The (classical) state process evolves independently and identically across time. The Tx knows the state either causally or non-causally, while the decoder remains ignorant. {A compact characterization of the currently known inner bounds follows.}

\begin{theorem}[QSTx capacities \cite{201604JPA_BocCaiNot}]
\label{Thm:QSTxKnownCapacities}
Consider a QSTx $(\CalX,\CalS,\ttp_S,\{\rho_{x,s}\}_{(x,s)},\kappa)$ as above, and fix a cost level $\tau$. Let $B$ denote the quantum output. If the state is known causally at the Tx and remains unknown at the Rx, then
\begin{equation}
\label{Eqn:CausalQSTxCapacity}
\Cc(\tau)
=
\max_{\substack{\ttp_U,~x=f(u,s):\\
\mathbb{E}\{\kappa(f(U,S),S)\}\leq \tau}}
\chi(U;B)_{\sigma}.
\end{equation}
where $U$ is independent of $S\sim \ttp_S$ and
\begin{equation}
\label{Eqn:CausalQSTxState}
\sigma^{UB}
=
\sum_{u}\ttp_U(u)\ketbra{u}{u}
\otimes
\sum_{s}\ttp_S(s)\rho_{f(u,s),s}.
\end{equation}
Moreover, one may restrict $|\CalU|\leq |\CalX||\CalS|$.

If the state is known non-causally at the Tx and remains unknown at the Rx, then
\begin{equation}
\label{Eqn:NonCausalQSTxCapacity}
C_{\mathrm{nc}}(\tau)
=
\lim_{n\rightarrow\infty}\frac{1}{n}
\max_{\substack{\ttp_{U^n|S^n},\,f}}
\bigl[
\chi(U^n;B^n)_{\sigma}
-
I(U^n;S^n)
\bigr],
\end{equation}
where the maximum is over conditional PMFs $\ttp_{U^n|S^n}$ and deterministic maps
$f:\CalU^n\times\CalS^n\to\CalX^n$ such that $\mathbb{E}\{\olinekappa(f(U^n,S^n),S^n)\}\leq \tau$. The state $\sigma^{U^nB^n}$ appearing above is given by
\begin{equation}
\label{Eqn:NonCausalQSTxState}
\begin{aligned}
\sigma^{U^nB^n}
&=
\sum_{u^n}\ttp_{U^n}(u^n)
\ketbra{u^n}{u^n}
\otimes \sigma_{u^n}^{B^n},\\
\sigma_{u^n}^{B^n}
&\define
\sum_{s^n}
\ttp_{S^n|U^n}(s^n|u^n)
\bigotimes_{t=1}^{n}
\rho_{f_t(u^n,s^n),s_t}.
\end{aligned}
\end{equation}
\end{theorem}
{The goal of the next section is to demonstrate that Rx being oblivious to the state results in strict penalties for non-commutative QSTx.}

\subsection{\textbf{A Noisy BB84-like QSTx}}
We now describe the QSTx of interest. Let the state and input alphabets be binary, and let the output states live in a qubit space, i.e. $\CalS=\CalX=\{0,1\}$ and $\CalH = \complex^{2}$. Starting from the BB84 \cite{198412ICCSSP_BenBra,201412TCS_BenBra} pure states and $\delta \in [0,\frac{1}{2}]$,
\begin{equation*}
\rho_{0,0}=\ketbra{0}{0},\qquad
\rho_{1,0}=\ketbra{1}{1},\qquad
\rho_{0,1}=\ketbra{+}{+},\qquad
\rho_{1,1}=\ketbra{-}{-}.
\end{equation*}
we define the noisy output states of the QSTx to be
\begin{equation}
\label{Eqn:NoisyStates}
\bar\rho^{(\delta)}_{x,s}
\define
(1-\delta)\rho_{x,s}+\delta\rho_{1-x,s},
\mbox{ for } (x,s)\in\{0,1\}^2.
\end{equation}
Input is Hamming cost constrained, i.e., $\mathbb{E}\{X\} \leq \tau$ where $\tau\in\left[0,\frac{1}{2}\right]$. A formal description is provided for ease of reference.
\begin{ex}
\label{Ex:QSTxExample}
Let $\CalS=\CalX=\{0,1\}$, $\CalH = \complex^{2}$, $S \sim \Ber(q)$, i.e., $\ttp_{S}(1)=q$ with $q \in [0,1]$. Referring to the above notation, we fix $\delta \in [0,\frac{1}{2}]$ and let $\bar\rho^{(\delta)}_{x,s} \define (1-\delta)\rho_{x,s}+\delta\rho_{1-x,s}$, with $(x,s)\in\{0,1\}^2$. The input is Hamming cost constrained, i.e., $\mathbb{E}\{X\} \leq \tau$ where $\tau\in\left[0,\frac{1}{2}\right]$.
\end{ex}

Firstly, note that if the state is available both at the Tx and Rx, then for each fixed $s$ the channel reduces to a binary symmetric channel of crossover probability $\delta$. Hence $
\Ced(\tau,\delta)=\hTwo(\tau*\delta)-\hTwo(\delta)$.

In our first finding, stated below, we provide an exact characterization of the natural single-letter non-causal rate for this channel, namely the \(n=1\) specialization of the general \(n\)-letter non-causal formula when the Rx is oblivious to the state, and show that in this example it coincides with the causal capacity. Toward that end, we define
\begin{align}
\label{Eqn:QSTxDefnR}
r(q)
&\define
\sqrt{q^2+(1-q)^2},\\
\label{Eqn:QSTxDefnDStar}
d_*(q,\tau)
&\define
\min\left\{
\sqrt{A^2+B^2}:
|A|\le q,\ |B|\le 1-q,\ A+B\ge 1-2\tau
\right\}.
\end{align}

We consider the one-letter non-causal Gelfand--Pinsker candidate
\begin{equation}
\label{Eqn:OneLetterCandidateGeneralq}
C_{\mathrm{nc}}^{\mathrm{sl}}(q,\tau,\delta)
\defeq
\max_{\substack{p(u|s),\,x(u,s):\\ \mathbb E[X]\le \tau}}
\bigl[\chi(U;B)-I(U;S)\bigr].
\end{equation}

\begin{proposition}
\label{Prop:QSTxCapacityExactExp} Consider Ex.~\ref{Ex:QSTxExample}.
For every $q\in[0,1]$, $\tau\in[0,1/2]$, and $\delta\in[0,1/2]$,
\begin{equation}
\label{Eqn:MainFormula}
\begin{aligned}
\Cc(q,\tau,\delta)=\Cnc(q,\tau,\delta)=
\hTwo\!\left(\frac{1+(1-2\delta)d_*(q,\tau)}{2}\right)-
\hTwo\!\left(\frac{1+(1-2\delta)r(q)}{2}\right),
\end{aligned}
\end{equation}
where $r(\cdot)$ and $d_{*}(\cdot,\cdot)$ are as in \eqref{Eqn:QSTxDefnR} and \eqref{Eqn:QSTxDefnDStar}. Moreover,
\begin{equation*}
d_*(q,\tau)=
\begin{cases}
\sqrt{q^2+(1-q-2\tau)^2},
& q\le \dfrac12,\ \tau<\dfrac12-q,\\[1.1ex]
\dfrac{1-2\tau}{\sqrt2},
& \left|q-\dfrac12\right|\le \tau\le \dfrac12,\\[1.1ex]
\sqrt{(q-2\tau)^2+(1-q)^2},
& q\ge \dfrac12,\ \tau<q-\dfrac12.
\end{cases}
\end{equation*}
\end{proposition}

\begin{proof}
In the first part, we consider a generic test channel and derive an upper bound on the achievable rate using the structure of the cq channel. In the second part, we identify a specific test channel achieving this upper bound.

We first write the channel outputs in Bloch form. Namely,
\begin{equation*}
\begin{gathered}
\bar\rho^{(\delta)}_{0,0}=\frac{I+(1-2\delta)Z}{2},\qquad
\bar\rho^{(\delta)}_{1,0}=\frac{I-(1-2\delta)Z}{2},\\
\bar\rho^{(\delta)}_{0,1}=\frac{I+(1-2\delta)X}{2},\qquad
\bar\rho^{(\delta)}_{1,1}=\frac{I-(1-2\delta)X}{2}.
\end{gathered}
\end{equation*}
Fix an arbitrary admissible test channel \(p(u|s)\) and a deterministic encoder \(x(u,s)\) satisfying \(\mathbb E[X]\le \tau\). For each \(u\), set
\[
q_u\defeq \Pr[S=1|U=u],
\qquad
\mathsf A_u\defeq 1-2x(u,1),
\qquad
\mathsf B_u\defeq 1-2x(u,0).
\]
Then \(\mathsf A_u,\mathsf B_u\in\{-1,+1\}\) and \(\sum_u p(u)q_u=q\). The conditional output state associated with \(U=u\) is
\begin{align}
\sigma_u
&\defeq
\sum_s p(s|u)\bar\rho^{(\delta)}_{x(u,s),s}
\nonumber\\
&=
\frac{1}{2}
\left(
I+(1-2\delta)
\left(\mathsf A_uq_uX+\mathsf B_u(1-q_u)Z\right)
\right).
\end{align}
Since \(XZ+ZX=0\), the square of
\[
\mathsf A_uq_uX+\mathsf B_u(1-q_u)Z
\]
is \(\bigl(q_u^2+(1-q_u)^2\bigr)I\). Hence the eigenvalues of \(\sigma_u\) are
\[
\frac{1\pm (1-2\delta)\sqrt{q_u^2+(1-q_u)^2}}{2},
\]
and therefore
\[
S(\sigma_u)
=
\hTwo\!\left(
\frac{1+(1-2\delta)\sqrt{q_u^2+(1-q_u)^2}}{2}
\right).
\]
Define
\[
g_\delta(t)
\defeq
\hTwo\!\left(
\frac{1+(1-2\delta)\sqrt{t^2+(1-t)^2}}{2}
\right)-\hTwo(t),
\qquad t\in[0,1].
\]
Since
\[
I(U;S)=\hTwo(q)-\sum_u p(u)\hTwo(q_u),
\]
we obtain
\begin{equation}
\label{Eqn:ProofGenericRateBound}
\chi(U;B)-I(U;S)
=
S(\bar\sigma)-\hTwo(q)-\sum_u p(u)g_\delta(q_u),
\end{equation}
where \(\bar\sigma\defeq \sum_u p(u)\sigma_u\). By Lemma~\ref{Lem:ConvexityGDelta} (see Appendix~\ref{App:AuxiliaryConvexityGeometry}), the function \(g_\delta\) is convex on \([0,1]\). Jensen's inequality then gives
\begin{equation}
\label{Eqn:JensenGDeltaUse}
\sum_u p(u)g_\delta(q_u)
\ge
g_\delta\!\left(\sum_u p(u)q_u\right)
=
g_\delta(q).
\end{equation}
It remains to upper-bound \(S(\bar\sigma)\). Set
\[
A\defeq \sum_u p(u)\mathsf A_uq_u,
\qquad
B\defeq \sum_u p(u)\mathsf B_u(1-q_u).
\]
Then
\[
\bar\sigma
=
\frac{1}{2}
\left(
I+(1-2\delta)(AX+BZ)
\right).
\]
As before, \((AX+BZ)^2=(A^2+B^2)I\), and hence
\[
S(\bar\sigma)
=
\hTwo\!\left(
\frac{1+(1-2\delta)\sqrt{A^2+B^2}}{2}
\right).
\]
We now relate \(A\) and \(B\) to the input-cost constraint. Since
\[
x(u,1)=\frac{1-\mathsf A_u}{2},
\qquad
x(u,0)=\frac{1-\mathsf B_u}{2},
\]
we have
\begin{align}
\mathbb E[X]
&=
\sum_u p(u)\bigl[q_u x(u,1)+(1-q_u)x(u,0)\bigr]
\nonumber\\
&=
\frac12
\sum_u p(u)\bigl(1-\mathsf A_uq_u-\mathsf B_u(1-q_u)\bigr)
\nonumber\\
&=
\frac12(1-A-B).
\end{align}
Thus \(\mathbb E[X]\le \tau\) implies \(A+B\ge 1-2\tau\). Moreover,
\[
|A|\le \sum_u p(u)q_u=q,
\qquad
|B|\le \sum_u p(u)(1-q_u)=1-q.
\]
Therefore \((A,B)\) belongs to the feasible set in the definition of \(d_*(q,\tau)\) in \eqref{Eqn:QSTxDefnDStar}, and consequently
\[
\sqrt{A^2+B^2}\ge d_*(q,\tau).
\]
Since the function \(x\mapsto \hTwo((1+(1-2\delta)x)/2)\) is decreasing on \([0,1]\), this yields
\begin{equation}
\label{Eqn:AverageEntropyDStarBound}
S(\bar\sigma)
\le
\hTwo\!\left(
\frac{1+(1-2\delta)d_*(q,\tau)}{2}
\right).
\end{equation}
Combining \eqref{Eqn:ProofGenericRateBound}, \eqref{Eqn:JensenGDeltaUse}, and \eqref{Eqn:AverageEntropyDStarBound}, and using the definition of \(g_\delta\), gives
\begin{align}
\chi(U;B)-I(U;S)
&\le
\hTwo\!\left(
\frac{1+(1-2\delta)d_*(q,\tau)}{2}
\right)
-\hTwo(q)-g_\delta(q)
\nonumber\\
&=
\hTwo\!\left(
\frac{1+(1-2\delta)d_*(q,\tau)}{2}
\right)
-
\hTwo\!\left(
\frac{1+(1-2\delta)r(q)}{2}
\right).
\end{align}
This proves the desired upper bound for every admissible non-causal test channel, and hence also for every causal test channel.

We now prove achievability. Let \((A^*,B^*)\) be a minimizer in the definition of \(d_*(q,\tau)\). For \(q\in(0,1)\), choose independent random signs \(\xi,\zeta\in\{-1,+1\}\) such that
\[
\mathbb E[\xi]=\frac{A^*}{q},
\qquad
\mathbb E[\zeta]=\frac{B^*}{1-q}.
\]
This is possible because \(|A^*|\le q\) and \(|B^*|\le 1-q\). Let \(U=(\xi,\zeta)\), independent of \(S\), and use the deterministic encoder
\[
x(U,1)=\frac{1-\xi}{2},
\qquad
x(U,0)=\frac{1-\zeta}{2}.
\]
The resulting average coefficients are exactly \(A^*\) and \(B^*\). Moreover, by the argument in Lemma~\ref{Lem:GeometryDStar}, the minimizing point satisfies \(A^*+B^*=1-2\tau\). Hence
\[
\mathbb E[X]=\frac12(1-A^*-B^*)=\tau.
\]
Since \(U\) is independent of \(S\), this test channel is already admissible for the causal formulation and \(I(U;S)=0\). For every value of \(U=(\xi,\zeta)\), the conditional output state has Bloch-vector norm
\[
(1-2\delta)\sqrt{q^2+(1-q)^2}=(1-2\delta)r(q),
\]
whereas the average output state has Bloch-vector norm
\[
(1-2\delta)\sqrt{(A^*)^2+(B^*)^2}=(1-2\delta)d_*(q,\tau).
\]
Therefore this causal test channel achieves
\[
\hTwo\!\left(
\frac{1+(1-2\delta)d_*(q,\tau)}{2}
\right)
-
\hTwo\!\left(
\frac{1+(1-2\delta)r(q)}{2}
\right).
\]
Together with the upper bound above and the inclusion of causal strategies among non-causal strategies, this proves
\[
\Cc(q,\tau,\delta)=\Cnc(q,\tau,\delta)
=
\hTwo\!\left(
\frac{1+(1-2\delta)d_*(q,\tau)}{2}
\right)
-
\hTwo\!\left(
\frac{1+(1-2\delta)r(q)}{2}
\right).
\]
The boundary cases \(q=0\) and \(q=1\) reduce to the ordinary binary symmetric channel with input-cost constraint, and the same expression is obtained by continuity. Finally, the explicit piecewise expression for \(d_*(q,\tau)\) is precisely Lemma~\ref{Lem:GeometryDStar} of Appendix~\ref{App:AuxiliaryConvexityGeometry}.
\end{proof}
A formal statement of rate loss for Ex.~\ref{Ex:QSTxExample} follows.

\begin{proposition}
\label{Prop:QSTXStrictRateLoss}
Assume $0<\tau\le \frac12$ and $0\le \delta<\frac12$. Then
\begin{equation}
\label{Eqn:StrictRateLossGeneralqWeak}
\Cnc(q,\tau,\delta)\le \Ced(\tau,\delta)
\qquad \text{for all }q\in[0,1].
\end{equation}
Moreover, as depicted in Fig.~\ref{Fig:FigPlotRateLossWrtState},\ref{Fig:FigPlotRateLossWrtCost}
\begin{equation}
\label{Eqn:StrictRateLossGeneralqStrict}
\Cnc(q,\tau,\delta)<\Ced(\tau,\delta)
\qquad \text{for every }q\in(0,1),
\end{equation}
whereas equality holds at the endpoints $q\in\{0,1\}$.
\end{proposition}

\begin{proof}
Set
\[
a\defeq 1-2\delta\in(0,1],
\qquad
c\defeq 1-2\tau\in[0,1),
\]
and define
\begin{equation}
\label{Eqn:JaDefinition}
J_a(x)\defeq h_2\!\left(\frac{1+ax}{2}\right),
\qquad x\in[0,1].
\end{equation}
By Proposition~\ref{Prop:QSTxCapacityExactExp},
\begin{equation}
\label{Eqn:CncAsJa}
C_{\mathrm{nc}}^{\mathrm{sl}}(q,\tau,\delta)
=
J_a(d_*(q,\tau))-J_a(r(q)).
\end{equation}
where
\[
r(q)=\sqrt{q^2+(1-q)^2},
\]
and
\[
d_*(q,\tau)=
\begin{cases}
\sqrt{q^2+(c-q)^2},
& q\le \dfrac12,\ \tau<\dfrac12-q,\\[2ex]
\dfrac{c}{\sqrt2},
& \left|q-\dfrac12\right|\le \tau\le \dfrac12,\\[2ex]
\sqrt{(q-c)^2+(1-q)^2},
& q\ge \dfrac12,\ \tau<q-\dfrac12.
\end{cases}
\]
Also,
\[
C_{\mathrm{ED}}(\tau,\delta)=h_2(\tau*\delta)-h_2(\delta).
\]
Since
\[
1-2(\tau*\delta)=ac,
\qquad
1-2\delta=a,
\]
the symmetry of the binary entropy function yields
\begin{equation}
\label{Eqn:CEDasJa}
C_{\mathrm{ED}}(\tau,\delta)=J_a(c)-J_a(1).
\end{equation}

\begin{figure}[t]
\centering
\begin{minipage}[t]{0.48\textwidth}
\centering
\includegraphics[width=\linewidth]{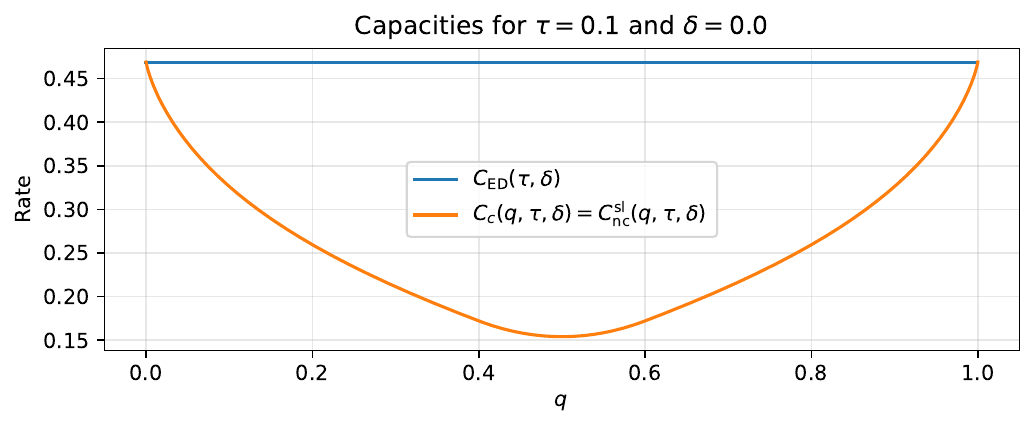}
\caption{Depiction of rate loss for a fixed cost constraint as a function of the state probability $q$.}
\label{Fig:FigPlotRateLossWrtState}
\end{minipage}
\hfill
\begin{minipage}[t]{0.48\textwidth}
\centering
\includegraphics[width=\linewidth]{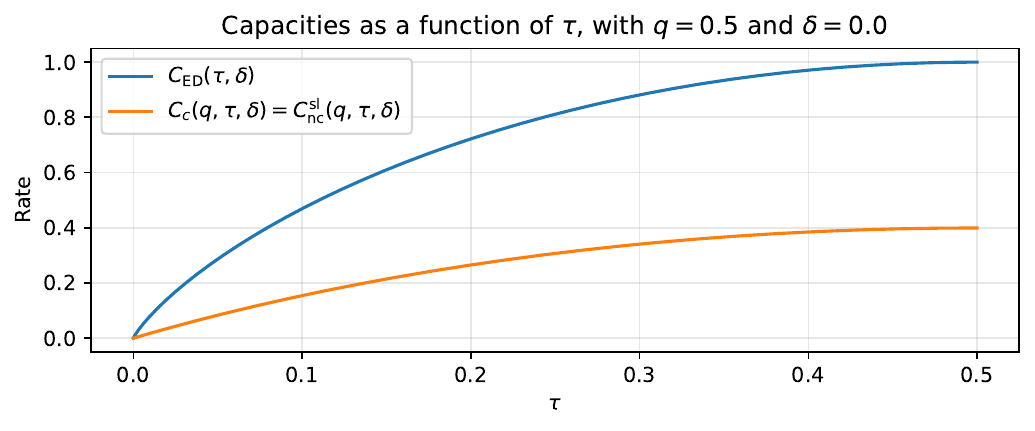}
\caption{Depiction of rate loss with respect to the cost constraint.}
\label{Fig:FigPlotRateLossWrtCost}
\end{minipage}
\end{figure}

We shall prove that the map
\[
q\longmapsto C_{\mathrm{nc}}^{\mathrm{sl}}(q,\tau,\delta)
\]
is symmetric with respect to \(q=\frac12\), and strictly decreasing on
\(\bigl[0,\frac12\bigr]\). This will imply that its maximum is attained
exactly at \(q=0\) and \(q=1\), where it coincides with
\(C_{\mathrm{ED}}(\tau,\delta)\).
We first note that
\[
r(1-q)=r(q),
\qquad
d_*(1-q,\tau)=d_*(q,\tau),
\]
by the explicit piecewise formula above. Hence
\begin{equation}
\label{Eqn:SymmetryHq}
C_{\mathrm{nc}}^{\mathrm{sl}}(1-q,\tau,\delta)
=
C_{\mathrm{nc}}^{\mathrm{sl}}(q,\tau,\delta).
\end{equation}
Therefore it suffices to study \(q\in[0,\frac12]\).
Define, for \(q\in[0,\frac12]\),
\[
H(q)\defeq J_a(d_*(q,\tau))-J_a(r(q)).
\]
We prove that \(H\) is strictly decreasing on \([0,\frac12]\). The derivative computations below are performed only at interior points \(q\in(0,\frac12)\). This is only relevant when \(a=1\), because \(J_1\) is not differentiable at \(x=1\); the only point where this boundary value can occur is \(q=0\), which will be evaluated separately by continuity.

We split into two cases.

\medskip

\noindent
\emph{Case 1: \(0\le q\le \frac12-\tau=\frac c2\).}
If \(c=0\), this case contains only the endpoint \(q=0\), so there is no derivative statement to prove in this case. Assume therefore, for the rest of Case~1, that \(c>0\). In this regime,
\[
d_*(q,\tau)=d(q)\defeq \sqrt{q^2+(c-q)^2},
\qquad
r(q)=\sqrt{q^2+(1-q)^2}.
\]
For \(q\in(0,\frac c2)\), the arguments \(d(q)\) and \(r(q)\) lie in \((0,1)\) whenever \(a=1\), and hence \(J_a\) is differentiable at the relevant points. We compute
\[
J_a'(x)
=
\frac{a}{2}\,h_2'\!\left(\frac{1+ax}{2}\right)
=
-\frac{a}{\ln 2}\,\artanh(ax).
\]
Therefore,
\[
H'(q)=J_a'(d(q))\,d'(q)-J_a'(r(q))\,r'(q).
\]
Now
\[
d'(q)=\frac{2q-c}{d(q)},
\qquad
r'(q)=\frac{2q-1}{r(q)},
\]
so
\begin{equation}
\label{Eqn:HprimeCase1}
H'(q)
=
\frac{a}{\ln 2}
\left[
\frac{c-2q}{d(q)}\,\artanh\!\bigl(a\,d(q)\bigr)
-
\frac{1-2q}{r(q)}\,\artanh\!\bigl(a\,r(q)\bigr)
\right].
\end{equation}
We now use the auxiliary function
\[
\psi_a(x)\defeq \frac{\artanh(ax)}{x}.
\]
On the range of \(d(q)\) and \(r(q)\) considered here, the function
\(x\mapsto \psi_a(x)\) is strictly increasing. Indeed, this follows from the
strict monotonicity of \(y\mapsto \artanh(y)/y\) on \((0,1)\); when \(a=1\),
the endpoint \(x=1\) is avoided because \(q>0\).
Because \(d(q)\le r(q)\), we obtain
\[
\frac{\artanh(a\,d(q))}{d(q)}
\le
\frac{\artanh(a\,r(q))}{r(q)}.
\]
Also,
\[
c-2q\le 1-2q,
\]
with strict inequality whenever \(\tau>0\). Combining these two facts in
\eqref{Eqn:HprimeCase1}, we get
\[
H'(q)\le 0
\qquad \text{for all }q\in\left(0,\frac c2\right).
\]
Moreover, since \(\tau>0\), we have \(c<1\), hence
\[
c-2q<1-2q
\]
for every \(q\in\left(0,\frac c2\right)\). Therefore
\begin{equation}
\label{Eqn:HprimeCase1Strict}
H'(q)<0
\qquad \text{for every }q\in\left(0,\frac c2\right).
\end{equation}

\medskip

\noindent
\emph{Case 2: \(\frac12-\tau=\frac c2\le q\le \frac12\).}
In this regime,
\[
d_*(q,\tau)=\frac{c}{\sqrt2}
\]
is constant, whereas \(r(q)=\sqrt{q^2+(1-q)^2}\). Hence
\[
H'(q)=-J_a'(r(q))\,r'(q).
\]
Using
\[
J_a'(r)=-\frac{a}{\ln 2}\,\artanh(ar)<0,
\qquad
r'(q)=\frac{2q-1}{r(q)}<0
\quad \text{for }q\in\left(0,\frac12\right),
\]
we conclude that
\begin{equation}
\label{Eqn:HprimeCase2Strict}
H'(q)<0
\qquad \text{for every }q\in\left(\frac c2,\frac12\right).
\end{equation}
Combining \eqref{Eqn:HprimeCase1Strict} and \eqref{Eqn:HprimeCase2Strict},
we conclude that \(H\) is strictly decreasing on \([0,\frac12]\). By the
symmetry \eqref{Eqn:SymmetryHq}, the function
\(q\mapsto C_{\mathrm{nc}}^{\mathrm{sl}}(q,\tau,\delta)\) is therefore
maximized exactly at \(q=0\) and \(q=1\), and is strictly smaller for every
\(q\in(0,1)\).

It remains to evaluate the endpoint value. At \(q=0\), we have
\[
r(0)=1,
\qquad
d_*(0,\tau)=c=1-2\tau,
\]
hence
\[
C_{\mathrm{nc}}^{\mathrm{sl}}(0,\tau,\delta)=J_a(c)-J_a(1).
\]
By \eqref{Eqn:CEDasJa}, this equals \(C_{\mathrm{ED}}(\tau,\delta)\). By
symmetry, the same holds at \(q=1\). This proves the proposition.
\end{proof}

\section{Efficiently Decoding Bivariate Interference via Coset Codes To Avoid Rate Loss over $3-$CQBCs}
\label{Sec:Broadcast}
As discussed in Sec.~\ref{Sec:Introduction}, the presence of a \textit{rate loss} motivates the design of a coding strategy for a broadcast channel (BC) that enables Rxs to decode as large a component of the interference as possible and minimizes the residual component that is precoded for by the Tx. Ex.~\ref{Ex:NonCommutative3CQBC} illustrates the relevance of this principle for the design of coding strategies over $3$-CQBCs. We use the standard notions of a $3$-CQBC and achievable rate triples, and focus here on the rate comparison needed for the separation between structured and unstructured coding strategies.
\begin{ex}
\label{Ex:NonCommutative3CQBC}
Let $\ulineCalX = \CalX_{1}\times \CalX_{2} \times \CalX_{3}$ denote the input set with $\CalX_{j} = \{0,1\}$ for $j \in [3]$. For $(x_{1},x_{2},x_{3}) \in \ulineCalX$, let
\begin{align}
\label{Eqn:NoisyCQBC}
\rho^{B_1Y_2Y_3}_{x_1x_2x_3}
&=
\bar\rho^{(\delta_1)}_{x_1,\,x_2\oplus x_3}
\otimes
\sigma^{Y_2}_{x_2}(\delta_2)
\otimes
\sigma^{Y_3}_{x_3}(\delta_3),\\
\label{Eqn:ClassicalBranches}
\sigma^{Y_k}_{x_k}(\delta_k)
&\define
(1-\delta_k)\ketbra{x_k}{x_k}+\delta_k\ketbra{1-x_k}{1-x_k},
\end{align}
for $k\in\{2,3\}$ and $\bar\rho^{(\delta_1)}_{x_1,\,x_2\oplus x_3}$ is defined based on \eqref{Eqn:NoisyStates}. Input bit $X_{1}$ is Hamming cost constrained to $\tau$, i.e. $\mathbb{E}[X_1]\le \tau$. What is the maximum rate of communication for Rx $1$, if Rxs $2$ and $3$ are fed at their respective capacities $1-h_{b}(\delta_{k})$ for $k=2,3$.
\end{ex}
We begin our study of Ex.~\ref{Ex:NonCommutative3CQBC} with a discussion and argue the sub-optimality of conventional IID code based strategies due to their inability to decode bivariate functions efficiently. We follow this by providing elements of a formal proof.

For simplicity, let's assume $\epsilon\define \delta_2=\delta_3$. Firstly, note that the Tx can input three bits $X_{1},X_{2},X_{3}$ with bit $X_{1}$ Hamming weight constrained to $\tau$. The three received components are unentangled. Moreover, received states at Rxs $2$ and $3$ are affected by only inputs bits $X_{2}$ and $X_{3}$ respectively. In fact, these are simple binary symmetric channels with crossover probabilities $\delta_{2},\delta_{3}$ respectively. The imposed rate requirements for Rxs $2$ and $3$ implies that input bits $X_{2},X_{3}$ must be uniform and independent. Let's henceforth focus on the Tx to Rx $1$ CQ channel. $X_2 \oplus X_3 \sim $ Ber$(\frac{1}{2})$  now plays the role of state $S$ in Ex.~\ref{Ex:QSTxExample}. Moreover, note that $X_2,X_3$ being dedicated to communicate to Rxs $2$ and $3$, Rx $1$ can be pumped information only through $X_1$ which is cost constrained.
The \textbf{key point} to note is that while the Hamming cost constraint on $X_1$ \textbf{constrains the possible inputs on bit $X_1$}, it does \textbf{not affect Rx $1$'s capability to decode}, with the latter governed only by the underlying Tx to Rx $1$ CQ channel.

From our study of Ex.~\ref{Ex:QSTxExample}, it is evident that if Rx $1$ has \textit{full} knowledge of $X_2 \oplus X_3$, then it can achieve a rate $\Ced(\tau,\delta_1)$, otherwise it has to be content with a lower rate. See Fig.~\ref{Fig:FigPlotRateLossWrtCost}. A key difference from Ex.~\ref{Ex:QSTxExample}, is that though $X_2,X_3$ are independent Ber$(\frac{1}{2})$, normalized entropy of $X_k^n$ is $1-\hTwo(\delta_{k})$ for $k=2,3$. What about normalized entropy of $X_2^n \oplus X_3^n$?

The answer crucially relies on the joint design of $X_2-,X_3-$codes. So long as two codes are unstructured IID and independent of each other, the number of pairwise additions explode resulting in normalized entropy of $X_2^n \oplus X_3^n$ being $\min\{1,2-2\hTwo(\epsilon)\}$. We conclude that \textbf{(Note 1)} if {unstructured IID coding strategy} must achieve a rate $\Ced(\tau,\delta_1)$ for Rx $1$, then Rx $1$ must be able to decode information at a total rate of $\min\{1,2-2\hTwo(\epsilon)\}+\Ced(\tau,\delta_1)$.

On the other hand, suppose that the \(X_2\)- and \(X_3\)-codebooks are
cosets of a common linear code, assuming for the moment that
\(\delta_2=\delta_3=\epsilon\). Since the sum of two cosets of a common
linear code is again a coset of the same linear code, the collection of
possible sums \(X_2^n\oplus X_3^n\) is constrained to a coset of rate
\(1-\hTwo(\epsilon)\). Thus receiver \(1\) does not need to decode the two
interfering codewords separately; it only needs to decode their sum, whose
effective rate is \(1-\hTwo(\epsilon)\).

Let
\[
        C_1:=\Ced(\tau,\delta_1)
        =
        \hTwo(\tau*\delta_1)-\hTwo(\delta_1),
\]
and let
\[
        C_{\mathrm{joint}}(\tau,\delta_1)
        :=
        I(X_1S;B_1)=\hTwo\!\left(
\frac{1+\frac{(1-2\delta_1)(1-2\tau)}{\sqrt2}}{2}
\right)
-
\hTwo(\delta_1),
\]
where \(S\sim\mathrm{Bern}(1/2)\), \(X_1\sim\mathrm{Bern}(\tau)\) are
independent, and the channel to receiver \(1\) is
\((x_1,s)\mapsto \bar\rho^{(\delta_1)}_{x_1,s}\). If
\[
        C_1+1-\hTwo(\epsilon)
        <
        C_{\mathrm{joint}}(\tau,\delta_1),
\]
then receiver \(1\) can jointly decode its own codeword and the structured
interference \(X_2^n\oplus X_3^n\). After recovering this interference, it is
in the state-informed situation and can communicate at rate \(C_1\).

In contrast, an unstructured IID strategy has no algebraic closure property
that compresses the bivariate interference into a codebook of rate
\(1-\hTwo(\epsilon)\). In the reduced unstructured comparison made below,
achieving the state-informed rate \(C_1\) forces receiver \(1\) to recover the
two interfering layers separately. Hence receiver \(1\) would have to support
a total rate \(C_1+2(1-\hTwo(\epsilon))\), which is impossible whenever
\[
        C_1+2(1-\hTwo(\epsilon))
        >
        C_{\mathrm{joint}}(\tau,\delta_1).
\]
Thus, whenever the two inequalities above hold simultaneously, the coset-code
strategy achieves the corner point
\[
        \bigl(C_1,1-\hTwo(\epsilon),1-\hTwo(\epsilon)\bigr),
\]
whereas the reduced unstructured IID subclass cannot.
\begin{theorem}
\label{Thm:StructuredCorner}
Assume
\begin{equation}
\label{Eqn:StructuredCondition}
\delta_2<\delta_3<\frac12
\qquad\text{and}\qquad
C_1+C_2<C_{\mathrm{joint}}(\tau,\delta_1).
\end{equation}
Then the rate triple
\begin{equation}
\label{Eqn:StructuredCorner}
(C_1,C_2,C_3)
=
\bigl(\Ced(\tau,\delta_1),\,1-\hTwo(\delta_2),\,1-\hTwo(\delta_3)\bigr)
\end{equation}
belongs to the closure of the structured inner bound obtained by specializing Theorem~1 of \cite{202503arXivBC_GouPad} to the channel in \eqref{Eqn:NoisyCQBC}. In particular, \eqref{Eqn:StructuredCorner} is achievable by a structured coding scheme based on nested coset codes.
\end{theorem}

\begin{proof}
We specialize Theorem~1 of Gouiaa--Padakandla with field size
\(\upsilon=2\). Let
\[
        U_2=U_3=\mathbb F_2,
\]
and for every \(\varepsilon\in(0,1/2)\), let
\begin{equation}
\label{Eqn:AuxChoiceUsers23}
        U_2,U_3 \sim \mathrm{Bern}\!\left(\frac12\right),
        \qquad
        V_2,V_3 \sim \mathrm{Bern}(\varepsilon),
\end{equation}
all independent. Let also
\begin{equation}
\label{Eqn:AuxChoiceUser1}
        V_1 \sim \mathrm{Bern}(\tau),
\end{equation}
independent of \((U_2,U_3,V_2,V_3)\). We choose the fusion map
\(f:U_2\times U_3\times V_1\times V_2\times V_3\to\{0,1\}^3\) as
\begin{equation}
\label{Eqn:FusionChoiceStructured}
        f(u_2,u_3,v_1,v_2,v_3)
        =
        (v_1,\;u_2\oplus v_2,\;u_3\oplus v_3).
\end{equation}
Equivalently,
\[
        X_1 = V_1,\qquad
        X_2 = U_2 \oplus V_2,\qquad
        X_3 = U_3 \oplus V_3.
\]
Thus the fusion map specifies the physical channel input, while the
structured interference decoded by receiver \(1\) is \(U_2\oplus U_3\).
Since \(V_1\sim\mathrm{Bern}(\tau)\), the cost constraint is met with equality.

Let
\[
        U:=U_2\oplus U_3,
        \qquad
        Z:=V_2\oplus V_3.
\]
Then the actual interference in the first branch is
\[
        X_2\oplus X_3=U\oplus Z,
\]
where
\[
        \Pr[Z=1]=2\varepsilon(1-\varepsilon)\xrightarrow[\varepsilon\downarrow0]{}0.
\]
Hence, as \(\varepsilon\downarrow0\), the channel seen by receiver \(1\)
converges to the cq channel
\[
        (v_1,u)\longmapsto \bar\rho^{(\delta_1)}_{v_1,u}.
\]
Consequently, by continuity of Holevo information,
\[
        I(V_1;B_1|U)\longrightarrow C_1,
        \qquad
        I(V_1U;B_1)\longrightarrow C_{\mathrm{joint}}(\tau,\delta_1).
\]
On the other hand, receivers \(2\) and \(3\) see BSC branches
\[
        Y_j=U_j\oplus V_j\oplus N_j,\qquad j=2,3,
\]
with \(V_j\sim\mathrm{Bern}(\varepsilon)\). Thus the perturbation \(V_j\)
carries a vanishing rate as \(\varepsilon\downarrow0\), while the structured
part \(U_j\) can be decoded at rates approaching
\[
        C_j=1-\hTwo(\delta_j).
\]

Therefore, if
\[
        C_1+\max\{C_2,C_3\}
        <
        C_{\mathrm{joint}}(\tau,\delta_1),
\]
the strict inequalities in Theorem~1 of Gouiaa--Padakandla can be satisfied
with arbitrarily small slack. Letting the slack and then \(\varepsilon\) tend
to zero yields the corner point
\[
        (C_1,C_2,C_3)
\]
in the closure of the structured inner bound.
\end{proof}

The condition \eqref{Eqn:StructuredCondition} is the natural noisy CQ analogue of the classical decoding threshold. In the present setting, however, the relevant quantity is not merely the interference-free capacity of the first branch, but the joint decoding quantity \(C_{\mathrm{joint}}(\tau,\delta_1)\) for simultaneous recovery of the first user’s codeword and the bivariate interference pattern. The nested-coset construction in Theorem~1 of \cite{202503arXivBC_GouPad}, specialized to the channel \eqref{Eqn:NoisyCQBC}, gives exactly this decoding condition and yields the corner point in \eqref{Eqn:StructuredCorner}.
\med{\textbf{A reduced unstructured subclass}} : We now make the comparison with unstructured IID coding explicit by isolating a concrete reduced subclass of the mixed Gouiaa--Padakandla inner bound. Recall that their general scheme contains a common layer $W$, six unstructured pairwise layers $Q_{ji}$, six structured layers $U_{ji}$, and three private layers $V_j$ \cite{202503arXivBC_GouPad}. The reduced subclass considered here is obtained by the specialization
$W\equiv \mathrm{const}$, $U_{ji}\equiv \mathrm{const}\ \text{for all }(j,i)$, $Q_{12}\equiv Q_{13}\equiv Q_{23}\equiv Q_{32}\equiv \mathrm{const}, Q_{21}\equiv T_2, Q_{31}\equiv T_3$.

This reduction is natural for \eqref{Eqn:NoisyCQBC}: users \(2\) and \(3\) each see an interference-free point-to-point branch, whereas receiver \(1\) alone faces the bivariate interference \(X_2\oplus X_3\); accordingly, within a purely IID comparison, the only pairwise unstructured layers that can naturally help manage this interference are \(Q_{21}\) and \(Q_{31}\), while the remaining \(Q\)-layers, the common layer \(W\), and the structured layers \(U_{ji}\) are suppressed in this reduced subclass. Thus the only nontrivial layers are $T_2,T_3,V_1,V_2,V_3$. We impose the factorization $p(t_2,t_3)\,p(v_1|t_2,t_3)\,p(v_2|t_2)\,p(v_3|t_3),$ and deterministic encoding functions $X_1=f_1(T_2,T_3,V_1),
X_2=f_2(T_2,V_2),
X_3=f_3(T_3,V_3).$
Receiver $1$ decodes $(T_2,T_3,V_1)$, receiver $2$ decodes $(T_2,V_2)$, and receiver $3$ decodes $(T_3,V_3)$. Let $\Cur(\tau)$ denote the closure of the set of all rate triples achievable by this reduced coding architecture under the cost constraint. For the proof below, we only need elementary packing constraints that are necessary for every point in $\Cur(\tau)$. Write $R_1=L_1, R_2=\nu_{21}+L_2, R_3=\nu_{31}+L_3$, where $\nu_{21}$ and $\nu_{31}$ denote the rates carried by $T_2$ and $T_3$ respectively, and $L_j$ is the private rate carried by $V_j$. Then every achievable triple in $\Cur(\tau)$ must satisfy
\begin{equation}
\label{Eqn:ReducedRx}
\begin{alignedat}{2}
L_1
&\le I(V_1;B_1|T_2,T_3),
&\qquad
L_2
&\le I(V_2;Y_2|T_2),\\
L_3
&\le I(V_3;Y_3|T_3),
&\qquad
\nu_{21}+L_2
&\le I(T_2,V_2;Y_2),\\
L_1+\nu_{21}+\nu_{31}
&\le I(V_1,T_2,T_3;B_1),
&\qquad
\nu_{31}+L_3
&\le I(T_3,V_3;Y_3).
\end{alignedat}
\end{equation}
These are single-user and full-sum packing constraints associated with the decoders $(T_2,T_3,V_1)$, $(T_2,V_2)$ and $(T_3,V_3)$.
\begin{proposition}
\label{Prop:ReducedUnstructuredSuboptimality}
Assume that $\delta_1<\frac12$ and
$\hTwo(\tau*\delta_1)+1>\hTwo(\delta_2)+\hTwo(\delta_3)$.
Then $(C_1,C_2,C_3)\notin \Cur(\tau)$.
\end{proposition}

\begin{proof}
Suppose, for contradiction, that $(C_1,C_2,C_3)\in \Cur(\tau)$. Since $Y_2$ depends only on $X_2$, we have the Markov chain $(T_2,V_2)-X_2-Y_2$, and therefore
\begin{equation}
\label{Eqn:R2Chain}
R_2=\nu_{21}+L_2\le I(T_2,V_2;Y_2)\le I(X_2;Y_2)\le C_2.
\end{equation}
Since $R_2=C_2$, all inequalities in \eqref{Eqn:R2Chain} are equalities. Thus $X_2\sim \Ber(1/2)$ and $X_2$ is a deterministic function of $(T_2,V_2)$. Likewise, $X_3\sim \Ber(1/2)$ and $X_3$ is a deterministic function of $(T_3,V_3)$. Define $T\define (T_2,T_3)$ and $S\define X_2\oplus X_3.$
By the reduced factorization $p(t_2,t_3)p(v_1|t_2,t_3)p(v_2|t_2)p(v_3|t_3)$, we have $V_2\perp V_3\mid T$, whence $X_2\perp X_3\mid T$.

Next, conditioned on $T=t$, receiver $1$ sees exactly the noisy BB84-like state-dependent channel $(x_1,s)\mapsto \bar\rho^{(\delta_1)}_{x_1,s}$
with state distribution $S\sim \Ber(q_t)$, where $q_t\define P(S=1|T=t)$, and conditional input cost $\tau_t\define P(X_1=1|T=t)$. Since $X_1=f_1(T,V_1)$ and $S$ is a function of $(T,V_2,V_3)$, we have $V_1\perp S\mid T$. Hence $I(V_1;B_1|T=t)\le \Cnc(q_t,\tau_t,\delta_1).$
Using $L_1 \le I(V_1;B_1|T_2,T_3),$ we obtain
\begin{equation}
\label{Eqn:R1Reduction}
R_1\le \sum_t p(t)\,\Cnc(q_t,\tau_t,\delta_1).
\end{equation}
On the other hand, revealing $S$ to the decoder yields for all $t$
\begin{equation}
\label{Eqn:PrimitivePointwiseED}
I(V_1;B_1|T=t)\le \hTwo(\tau_t*\delta_1)-\hTwo(\delta_1).
\end{equation}
Let $F_{\delta_1}(\lambda)\define \hTwo(\lambda*\delta_1)-\hTwo(\delta_1)$ for $\lambda\in[0,1]$. Since $\delta_1<1/2$, the map $F_{\delta_1}$ is strictly concave on $[0,1]$. Summing \eqref{Eqn:PrimitivePointwiseED} and applying Jensen's inequality, we obtain
\begin{equation}
\label{Eqn:R1UpperBound}
R_1\le \sum_t p(t)F_{\delta_1}(\tau_t)
\le F_{\delta_1}\!\left(\sum_t p(t)\tau_t\right)
\le F_{\delta_1}(\tau)=C_1.
\end{equation}
Because $R_1=C_1$, equality must hold throughout \eqref{Eqn:R1UpperBound}. Hence $\tau_t=\tau$ for $p(t)$-a.e. $t$. Returning to \eqref{Eqn:R1Reduction} and invoking Proposition~\ref{Prop:QSTXStrictRateLoss}, we conclude that equality can hold only if
$q_t\in\{0,1\}$ for $p(t)$-a.e. $t$.
Equivalently, $H(S|T)=0$. Combined with $X_2\perp X_3\mid T$, this forces $H(X_2|T)=H(X_3|T)=0$: indeed, if $X_2\oplus X_3$ is deterministic given $T=t$ while $X_2$ and $X_3$ are conditionally independent given $T=t$, then the two conditional distributions cannot both be non-degenerate. Moreover, under the reduced factorization, $X_2=f_2(T_2,V_2)$ is conditionally independent of $T_3$ given $T_2$, and $X_3=f_3(T_3,V_3)$ is conditionally independent of $T_2$ given $T_3$; hence $H(X_2|T_2)=H(X_3|T_3)=0$. Consequently $I(V_2;Y_2|T_2)=I(V_3;Y_3|T_3)=0$, so the second and fifth inequalities in \eqref{Eqn:ReducedRx} imply $L_2=L_3=0$. Therefore the entire rates $R_2$ and $R_3$ are carried by $T_2$ and $T_3$, both decoded by receiver $1$. We may now use the receiver-$1$ sum-rate inequality in \eqref{Eqn:ReducedRx}: $R_1+R_2+R_3=L_1+\nu_{21}+\nu_{31}\le I(V_1,T_2,T_3;B_1)$.

Since $B_1$ depends on the channel input only through $(X_1,S)$, data processing yields $R_1+R_2+R_3\le I(X_1,S;B_1).$ Each of the four states $\bar\rho^{(\delta_1)}_{x,s}$ has entropy $\hTwo(\delta_1)$, while the output system $B_1$ is a qubit. Hence $I(X_1,S;B_1)\le 1-\hTwo(\delta_1).$ Combining the last two equations, $R_1+R_2+R_3\le 1-\hTwo(\delta_1).$ At the corner point $(C_1,C_2,C_3)$, the left-hand side equals $\hTwo(\tau*\delta_1)-\hTwo(\delta_1)+1-\hTwo(\delta_2)+1-\hTwo(\delta_3),$ so that $\hTwo(\tau*\delta_1)+1\le \hTwo(\delta_2)+\hTwo(\delta_3),$ a contradiction. This proves the proposition.
\end{proof}
Theorem~\ref{Thm:StructuredCorner} and Proposition~\ref{Prop:ReducedUnstructuredSuboptimality} together show that the noisy BB84-like state-dependent channel furnishes a natural CQBC example where a structured coding strategy strictly outperforms a natural reduced unstructured subclass, as made explicit by the following corollary.

\begin{corollary}
\label{Cor:NonemptySeparationRegime}
The hypotheses of Theorem~\ref{Thm:StructuredCorner} and
Proposition~\ref{Prop:ReducedUnstructuredSuboptimality} are simultaneously
satisfiable. In particular, the choice $\tau=0.05,
\delta_1=0,
\delta_2=0.158,
\delta_3=0.159$ satisfies both sets of conditions.
Consequently, for these parameters, the corner point $(C_1,C_2,C_3)$
belongs to the closure of the structured inner bound of
Theorem~\ref{Thm:StructuredCorner}, but does not belong to the reduced
unstructured subclass \(\mathcal R_{\mathrm{u,red}}(\tau)\).
\end{corollary}

\appendices

\section{Auxiliary Convexity and Geometry Lemmas}
\label{App:AuxiliaryConvexityGeometry}

\subsection{Convexity}

\begin{lemma}
\label{Lem:ConvexityGDelta}
For every \(\delta\in[0,1/2]\), the function
\[
g_\delta(t)
=
\hTwo\!\left(
\frac{1+(1-2\delta)\sqrt{t^2+(1-t)^2}}{2}
\right)-\hTwo(t),
\qquad t\in[0,1],
\]
is convex on \([0,1]\).
\end{lemma}

\begin{proof}
Set
\[
a\defeq 1-2\delta\in[0,1],
\qquad
r(t)\defeq \sqrt{t^2+(1-t)^2},
\qquad
\lambda_a(t)\defeq \frac{1+a\,r(t)}{2}.
\]
Then
\[
g_\delta(t)=\hTwo(\lambda_a(t))-\hTwo(t).
\]
For \(t\in(0,1)\),
\[
r'(t)=\frac{2t-1}{r(t)},
\qquad
r''(t)=\frac{1}{r(t)^3},
\]
hence
\[
\lambda_a'(t)=\frac{a(2t-1)}{2r(t)},
\qquad
\lambda_a''(t)=\frac{a}{2r(t)^3}.
\]
Recall that
\[
\hTwo'(x)=\log_2\frac{1-x}{x},
\qquad
\hTwo''(x)=-\frac{1}{x(1-x)\ln 2}.
\]
Therefore,
\[
g_\delta''(t)
=
\hTwo''(\lambda_a(t))(\lambda_a'(t))^2
+
\hTwo'(\lambda_a(t))\lambda_a''(t)
-
\hTwo''(t).
\]
Using
\[
\lambda_a(t)(1-\lambda_a(t))=\frac{1-a^2r(t)^2}{4},
\qquad
t(1-t)=\frac{1-r(t)^2}{2},
\qquad
(2t-1)^2=2r(t)^2-1,
\]
a direct simplification yields
\[
g_\delta''(t)
=
\frac{1}{\ln 2}
\left[
\frac{2}{1-r^2}
-
\frac{a^2(2r^2-1)}{r^2(1-a^2r^2)}
+
\frac{a}{2r^3}\ln\frac{1-ar}{1+ar}
\right],
\]
where \(r=r(t)\). Now fix \(r\in[1/\sqrt2,1)\), and set
\[
x\defeq ar\in[0,r].
\]
Define
\[
H_r(x)
\defeq
-\frac{x^2(2r^2-1)}{1-x^2}
+
\frac{x}{2}\ln\frac{1-x}{1+x}.
\]
Then
\[
g_\delta''(t)
=
\frac{1}{\ln 2}
\left[
\frac{2}{1-r^2}+\frac{1}{r^4}H_r(x)
\right].
\]
A direct differentiation gives
\[
H_r'(x)
=
\frac12\ln\frac{1-x}{1+x}
-
\frac{x}{1-x^2}
-
\frac{2x(2r^2-1)}{(1-x^2)^2}.
\]
Since \(x\in(0,1)\) and \(r\in[1/\sqrt2,1)\), each term on the right-hand side is nonpositive, hence
\[
H_r'(x)\le 0.
\]
Therefore \(H_r\) is decreasing, so for fixed \(t\), the quantity \(g_\delta''(t)\) is decreasing in \(a\in[0,1]\). Consequently,
\[
g_\delta''(t)\ge g_0''(t),
\]
where \(g_0\) corresponds to \(a=1\), i.e., \(\delta=0\).
For \(\delta=0\), one has
\[
g_0''(t)
=
\frac{1}{2\ln 2\,r(t)^3}
\left(
\frac{2r(t)}{1-r(t)^2}
-
\ln\frac{1+r(t)}{1-r(t)}
\right).
\]
Define, for \(0\le s<1\),
\[
\phi(s)\defeq \frac{2s}{1-s^2}-\ln\frac{1+s}{1-s}.
\]
Then \(\phi(0)=0\) and
\[
\phi'(s)=\frac{4s^2}{(1-s^2)^2}\ge 0.
\]
Hence \(\phi(s)\ge 0\) for all \(s\in[0,1)\). Applying this with \(s=r(t)\), we obtain
\[
g_0''(t)\ge 0.
\]
Therefore
\[
g_\delta''(t)\ge 0
\qquad\text{for all }t\in(0,1).
\]
Since \(g_\delta\) is continuous on \([0,1]\), it is convex on \([0,1]\).
\end{proof}

\subsection{Geometry of the Minimization}

\begin{lemma}
\label{Lem:GeometryDStar}
For every \(q\in[0,1]\) and \(\tau\in[0,1/2]\),
\[
d_*(q,\tau)=
\begin{cases}
\sqrt{q^2+(1-q-2\tau)^2},
& q\le \dfrac12,\ \tau<\dfrac12-q,\\[2ex]
\dfrac{1-2\tau}{\sqrt2},
& \left|q-\dfrac12\right|\le \tau\le \dfrac12,\\[2ex]
\sqrt{(q-2\tau)^2+(1-q)^2},
& q\ge \dfrac12,\ \tau<q-\dfrac12.
\end{cases}
\]
\end{lemma}

\begin{proof}
Set
\[
c\defeq 1-2\tau\in[0,1].
\]
We must minimize \(A^2+B^2\) over
\[
|A|\le q,
\qquad
|B|\le 1-q,
\qquad
A+B\ge c.
\]
We first show that any minimizer must satisfy \(A+B=c\). If \(c=0\), then the feasible point \((0,0)\) has squared norm \(0\), and hence every minimizer satisfies \(A=B=0\), so in particular \(A+B=c\). We may therefore assume \(c>0\). Let \((A,B)\) be feasible with \(A+B>c\), and define
\[
\lambda\defeq \frac{c}{A+B}\in(0,1).
\]
Then \((\lambda A,\lambda B)\) is still feasible, because
\[
|\lambda A|\le |A|\le q,
\qquad
|\lambda B|\le |B|\le 1-q,
\qquad
\lambda A+\lambda B=c,
\]
and its squared norm is strictly smaller:
\[
(\lambda A)^2+(\lambda B)^2=\lambda^2(A^2+B^2)<A^2+B^2.
\]
Therefore no minimizer can satisfy \(A+B>c\).
Hence we minimize
\[
A^2+(c-A)^2
\]
over the intersection of the line \(A+B=c\) with the rectangle
\[
[-q,q]\times[-(1-q),1-q].
\]
Without box constraints, the minimum is attained at
\[
A=B=\frac c2=\frac{1-2\tau}{2}.
\]
This point lies in the rectangle if and only if
\[
\frac{1-2\tau}{2}\le q
\qquad\text{and}\qquad
\frac{1-2\tau}{2}\le 1-q,
\]
that is,
\[
\left|q-\frac12\right|\le \tau.
\]
In this case,
\[
d_*(q,\tau)=\frac{1-2\tau}{\sqrt2}.
\]
Suppose now that \(q\le 1/2\) and \(\tau<1/2-q\). Then
\[
\frac{1-2\tau}{2}>q,
\]
so the unconstrained minimizer lies to the right of the admissible segment. Hence the constrained minimizer is obtained by saturating \(A=q\). Since \(A+B=c\), this gives
\[
B=c-q=1-q-2\tau.
\]
Therefore
\[
d_*(q,\tau)=\sqrt{q^2+(1-q-2\tau)^2}.
\]
Finally, suppose that \(q\ge 1/2\) and \(\tau<q-1/2\). Then
\[
\frac{1-2\tau}{2}>1-q,
\]
so the unconstrained minimizer lies above the admissible segment. Hence the constrained minimizer is obtained by saturating \(B=1-q\). Since \(A+B=c\), this gives
\[
A=c-(1-q)=q-2\tau.
\]
Therefore
\[
d_*(q,\tau)=\sqrt{(q-2\tau)^2+(1-q)^2}.
\]
This proves the piecewise expression for \(d_*(q,\tau)\).
\end{proof}

\newpage

\bibliographystyle{sty/IEEEtran}
\bibliography{bibliography/CosetCdsFor3CQChnls}

\end{document}